\documentclass[a4paper,11pt]{IEEEtran}

\usepackage{amsfonts}
\usepackage{amsbsy}
\usepackage{amssymb}
\usepackage{amsmath}
\usepackage{algorithm}
\usepackage{algpseudocode}
\usepackage[pdftex]{graphicx}
\usepackage{subfig}
\usepackage{color}
\usepackage{amsthm}
\graphicspath{{C:/Users/zhanghant/Dropbox/Zhanghan Journal Paper 2017/arXiv_submitted}}
\makeatletter
\renewcommand*\env@matrix[1][*\c@MaxMatrixCols c]{%
	\hskip -\arraycolsep
	\let\@ifnextchar\new@ifnextchar
	\array{#1}}
\makeatother

\newtheorem{definition}{Definition}[section]

\newtheorem{thm}{Theorem}[section]
\newtheorem{cor}[thm]{Corollary}

\newtheorem{rem}[thm]{Remark}
\newcommand{\bmat}{\left[ \begin{array}}
\newcommand{\emat}{\end{array} \right]}

\newcommand{\DEG}{\mathrm{deg}\;}

\DeclareMathOperator{\e}{e}

\title{\LARGE \bf Linear system security--- detection and correction of adversarial attacks in the noise-free case}
\author{Zhanghan Tang, Margreta Kuijper, Michelle Chong, Iven Mareels and Chris Leckie\footnote{M.S. Chong is with Lund University, Sweden; all other authors are with The University of Melbourne, Australia} \\ \today}

\begin{document}

\maketitle
%%%%%%%%%%%%%%%%%%%%%%%%%%%%%%%%%%%%%%%%%%%%%%%%%%%%%%%%%%%%%%%%%%%%%%%%%%%%%%%
 \begin{abstract}
 We address the problem of attack detection and attack correction for multi-output discrete-time linear time-invariant systems under sensor attack. More specifically, we focus on the situation where adversarial attack signals are added to some of the system's output signals. A `security index' is defined to characterize the vulnerability of a system against such sensor attacks. Methods to compute the security index are presented as are algorithms to detect and correct for sensor attacks. The results are illustrated by examples involving multiple sensors.
 \end{abstract}
% \ruimte
% {\bf Mathematics Subject Classification (2010):} ...
%%%%%%%%%%%%%%%%%%%%%%%%%%%%%%%%%%%%%%%%%%%%%%%%%%%
\section{Introduction}\label{sec_intro}
\quad In today's society, the physical infrastructure in support of critical services such as water and energy can be described as a cyber-physical system. The end-to-end service can therefore be affected not only by the loss of functionality of the physical assets (breaking of wires or pipes, loss of motors or sensors, due to, for example, accidental breakage or maintenance outages or due to a natural disaster) but also through loss of functionality in the cyber assets (like loss of bandwidth in communication channels, message loss, or a virus in a computer operating system).

In this paper, the particular cyber attack scenario where sensor signals may be corrupted by additive
signals is considered. In the case where such signals are injected based on knowledge of a model of the system's
behaviour, it no longer suffices to treat these external signals as mere disturbances or noise, as in fact they can be used to control the behaviour. For instance, it is conceivable that a modern autonomous vehicle may be hijacked using ’sensor spoofing’ \cite{warner16}.

This issue has captured the attention of the control community. Several attack detection methods have been proposed in the literature, such as \cite{pasqualettiDB13,chenKarMouraICASSP15,kmanandhar14}. Correction methods have also been discussed in the literature, such as \cite{fawziTD14,shoukryT2016,mpajic15,ChongWakaikiHespanhaACC15}.

In this paper, the case of discrete-time, linear and
time-invariant (LTI) systems where output signals may
be compromised is considered. Attack signals are
modeled as signals added to the output signals. To
assist in the analysis, the notion of the `security index' of a system is introduced. This is analogous to the notion of the `minimum distance' used in coding theory. The main tools used in this paper are kernel representations of systems and the setting is the behavioral approach \cite{poldermanW97}. The security index is a quantitative representation-free measure of
the vulnerability of a system to sensor attacks. It speaks to the detectability and correctability of attack signals.

Unlike much of the work in this area, our starting point in this paper is a kernel representation  (see for example ~\cite[Ch2.5]{poldermanW97}) rather than a state space representation. Reasons for this are: 1. Many well-established theorems based on kernel representations can be applied when we are discussing a system using a behavioral approach. 2. Every kernel representation can be brought into state space form (see subsection \ref{sec:attack_detect}.4) and every observable system in state space form can be transformed into a kernel representation. 3. The implementation of systems in kernel representation can be done straightforwardly using shift operators.

Previous works involving systems under sensor attacks include \cite{ChongWakaikiHespanhaACC15,fawziTD14,chenKarMouraICASSP15,shoukryT2016,Frisk99}. The work of~\cite{fawziTD14} focuses on reconstructing the state value using a relaxed optimization program to approximate an NP-hard $l_0$-norm optimization problem. We demonstrate that our formulation simplifies the approach. There is a consensus of our result with \cite{ChongWakaikiHespanhaACC15} that the output signals are only guaranteed to be reconstructible if a certain upper bound on the number of attacked sensors is met.
%strictly less than half of the sensors are compromised. 
We reformulate the assumptions of \cite{ChongWakaikiHespanhaACC15} in terms of the 
%problem under the scope of 
security index and derive methods for detection as well as correction. Other related works are \cite{Teix15} which focuses on the establishment of the models for various attack signals; and  \cite{sandbergTJ2010} which focuses on the security of power networks.

The focus of our paper is on the development of a conceptual approach to attack detection and correction. Unlike e.g. \cite{ChongWakaikiHespanhaACC15,leeSE2015ECC,Zhang2015a}, we restrict ourselves to a noise-free environment to enable the reader to understand the essence of the proposed methods. These methods then serve as a starting point for further research on the noisy case.

The outline of this paper is as follows. Section \ref{sec:notation} presents some notation used in the paper. Section \ref{sec:prob_statement} describes the system and the problem statements, while Section \ref{sec:sec_index} defines the security index, attack detectability and attack correctability. A method of computing the security index in terms of a kernel representation is given in Section \ref{sec:computation}. Section \ref{sec:sim_tra} presents the Kronecker-Hermite canonical form representation. Section \ref{sec:attack_detect} gives an attack detection method and uses the Kronecker-Hermite canonical form to design attack correction methods, first for the maximally secure case, then for the general case. The theory is illustrated in Section \ref{sec:example} by results and simulations for two discrete-time LTI system examples that involve multiple sensors. Finally, conclusions and future work directions are presented in Section \ref{sec:con}. This paper builds on preliminary work by two of its authors~\cite{chongK2016mtns,Chong2016cdc}, and the main new result is the attack correction method.
\section{Notation}\label{sec:notation} 
\begin{itemize}
	\item Let $\mathbb{Z}_{+}=\{0,1,\dots\}$ and $\mathbb{R}:=(-\infty,\infty)$. 
	\item  The $N$-dimensional signal $y = (y_1,...,y_N)^T$ is denoted as $y:\mathbb{Z}_+\rightarrow\mathbb{R}^N$. 
	%\item The trajectory of an $N$ dimension signal $(y_1,...,y_n)^T$ is denoted by $y: \mathbb{Z}_{+}\rightarrow\mathbb{R}^N$.
	\item An $N\times N$ identity matrix is denoted by $\mathbb{I}_N$.
	%\item The kernel of a matrix $G$ is denoted by $\textrm{ker}(G):=\{v:Gv=0 \}$.
	\item The support of a signal is denoted by $\textrm{supp }(y):=\{i\in\{1,\dots,N\}: y_i \mbox{ is not the zero signal} \}$. 
	\item The weight of a signal $y$ is denoted by $\|y\|:=|\textrm{supp}(y)|$, i.e., the number of components of $y$ that are non-zero signals.
	\item If $\mathcal{J}$ is a subset of $\{1,\dots,N\}$ then its complementary set is denoted by $\bar{\mathcal{J}}$.
	%\item The determinant of a matrix $G$ is denoted by $\textrm{det}(G)$.
	\item The shift operator $\sigma$ is defined as $\sigma y(t) :=y(t+1)$.
	\item The degree of a polynomial $a(\xi)$ is denoted by $\DEG a(\xi)$.
	\item The greatest common divisor of two polynomials $a(\xi)$ and $b(\xi)$ is denoted by $\text{GCD}(a(\xi),b(\xi))$.
\end{itemize}

%%%%%%%%%%%%%%%%%%%%%%%%%%%%%%%%%%%%%%
\section{Problem statement}\label{sec:prob_statement}
Consider a linear time-invariant (LTI) system $\Sigma$ in its kernel representation as follows
\begin{align}  \label{eq:system}
\Sigma: \quad R(\sigma)y = 0,
\end{align}
where $y:\mathbb{Z}_+\rightarrow\mathbb{R}^N$ is the sensor output signal of the system $\Sigma$ and $R(\xi)$ is a real polynomial matrix of full rank, meaning that the system's behaviour is autonomous with no free variables. The size of $R(\xi)$ is $N\times N$.

\begin{definition}[Behaviour of the system $\Sigma$]\label{def:be}
	The behaviour of the system $\Sigma$ is defined as the set given by
	\begin{equation}\label{eq:be_def}
	\mathcal{B} = \{y:\mathbb{Z}_+\rightarrow\mathbb{R}^N\;|\; R(\sigma)y = 0\}.
	\end{equation}
\end{definition}

Consider a class $\mathcal{A}$ of attack signals $\eta:\mathbb{Z}_+\rightarrow\mathbb{R}^N$. A corrupted output signal is $r = y+\eta$. Here we denote the resulting system by $\Sigma_\mathcal{A}$, more specifically we have the following definition.

\begin{definition}[Behaviour of the system $\Sigma_\mathcal{A}$]\label{def:be_attack}
	The behaviour of the corrupted system $\Sigma_\mathcal{A}$ is defined as the set of possible received signals
	\begin{equation}\label{eq:system_attacked}
	\mathcal{B}_\mathcal{A} = \{r:\mathbb{Z}_+\rightarrow\mathbb{R}^N\; |\; r = y+\eta,\mbox{ where }y\in\mathcal{B},\eta\in\mathcal{A}\}.
	\end{equation}
\end{definition}

\begin{definition}[Attack detectability] \label{def:detect} A non-zero attack signal $\eta\in\mathcal{A}$ is detectable if $\eta\notin\mathcal{B}$.
\end{definition}

\begin{definition}[Attack correctability] \label{def:correct} A non-zero attack signal $\eta\in\mathcal{A}$ is correctable if for all $\eta'\neq\eta$, the following is satisfied
	\begin{equation}
	\eta'\in\mathcal{A}\Rightarrow\eta-\eta'\notin\mathcal{B}.
	\end{equation}
\end{definition}
Our objectives in this paper are to first determine the feasibility of attack detection/correction and then to give an attack detection as well as correction method. We show that these methods are guaranteed to produce the correct outcome under certain assumptions about the attack set $\mathcal{A}$.
%%%%%%%%%%%%%%%%%%%%%%%%%%%%%%%%%%%%%%
% \subsection{Open-loop sensor attacks}
% essentially sufficient to look at systems without inputs because we know the inputs
% \subsection{Open-loop actuator attacks}
% essentially sufficient to look at systems without outputs because we observe the outputs

%%%%%%%%%%%%%%%%%%%%%%%%%%%%%%%%%%%%%%

\section{Attack detection/correction feasibility}\label{sec:sec_index}
In this section we first address the vulnerability of a system given by \eqref{eq:system} against attacks on its sensor outputs $y$. We then introduce a concept that is central to this paper called the \emph{security index} $\delta(\Sigma)$ of the system $\Sigma$, and then we state conditions to achieve attack detectability and correctability. These conditions are stated in terms of $\delta(\Sigma)$. The definitions and results of this section can also be found in~\cite{chongK2016mtns,Chong2016cdc}.
%We assume $\mathcal{A} = \{\eta|\|\eta\|\leq N',\eta\neq 0\}$ where $N'\leq N\in\mathbb{Z}_+$.

\begin{definition}\label{def:security_index}
	The security index of the system $\Sigma$ is defined as
	\begin{equation}\label{eq:security_index}
	\delta(\Sigma) := \underset{0\neq y\in\mathcal{B}}{\min}{\|y\|}.
	\end{equation}	
\end{definition}

\begin{thm} \label{Theorem:detect}\textbf{\textup{(Attack detection capability of the system)}} Let $\mathcal{A}=\{\eta:\mathbb{Z}_+\rightarrow\mathbb{R}^N\; |\; \|\eta\|<\delta(\Sigma) \text{and } \eta\neq0\}$. All attack signals $\eta\in\mathcal{A}$ are detectable.
\end{thm}
\begin{proof}
	For any attack signal $\eta$ from $\mathcal{A}$ we must have $\eta\notin\mathcal{B}$ because of Definition \ref{def:security_index}. According to Definition \ref{def:detect}, $\eta$ is then detectable and this completes the proof. 
\end{proof}
Because of the above theorem, the security index $\delta(\Sigma)$ of a system $\Sigma$ can be viewed as the minimum number of sensors that have to be attacked in order to implement an undetectable attack. For example, if a system $\Sigma$ with $N=7$ sensors has security index $\delta(\Sigma) = 5$, then at least 5 of the 7 sensors have to be attacked to achieve an undetectable attack. In accordance with~\cite{chongK2016mtns,Chong2016cdc} we call a system $\Sigma$ with $N$ outputs \emph{maximally secure} if its security index equals $\delta(\Sigma)=N$.

\begin{thm}\label{Theorem:Correctability}\textbf{\textup{(Attack correction capability of the system)}} Let $\mathcal{A} = \{\eta:\mathbb{Z}_+\rightarrow\mathbb{R}^N\;|\; \|\eta\|<\delta(\Sigma)/2 \text{and } \eta\neq 0\}$. All attack signals $\eta\in\mathcal{A}$ are correctable.
\end{thm}
\begin{proof}
	Consider an attack signal $\eta$ from $\mathcal{A}$. If there exists another non-zero $\eta'$ with $\|\eta'\|<\delta(\Sigma)/2$ and $\eta'\neq\eta$, then $\|\eta-\eta'\|\leq\|\eta\|+\|\eta'\|<\delta(\Sigma)$ which implies $\eta-\eta'\notin\mathcal{B}$. According to Definition \ref{def:correct}, $\eta$ is then correctable and this completes the proof.
\end{proof}

It follows from the above theorem that the value of $\delta(\Sigma)/2$ can be viewed as the minimum number of sensors that have to be attacked in order to implement an uncorrectable attack. For example, if a system $\Sigma$ with $N=7$ sensors has security index $\delta(\Sigma) = 5$, then at least 3 of the 7 sensors have to be attacked to achieve an uncorrectable attack.

\begin{rem}
	Regarding the detection and correction of an attack signal $\eta$, there are three different situations: 
	
	Case 1: $\|\eta\|<\delta(\Sigma)/2$, i.e., the attack signal is both detectable and correctable. This means that one can detect the  existence of $\eta$ and there exists a unique correction for this attack signal.
	
	Case 2: $\delta(\Sigma)/2\leq\|\eta\|<\delta(\Sigma)$, i.e., the attack signal is detectable but not correctable. This means that one can detect the existence of $\eta$ but cannot guarantee the existence of a unique correction for this attack signal.
	
	Case 3: $\delta(\Sigma)\leq\|\eta\|$, i.e., the attack signal is neither detectable nor correctable. This means that one cannot guarantee that the attack signal can be detected, also one cannot guarantee the existence of a unique correction for this attack signal.
\end{rem}
%%%%%%%%%%%%%%
\section{Computation of the security index}\label{sec:computation}
In this section, we seek to express the security index $\delta(\Sigma)$ of a LTI system $\Sigma$ in terms of its kernel representation \eqref{eq:system}. The following preliminaries on polynomial matrices are required.

A square polynomial matrix is called \emph{unimodular} if it has a polynomial inverse and a non-square polynomial matrix is called \emph{left (right) unimodular} if it has a polynomial left (right) inverse. Two polynomial matrices $R(\xi)$ and $Q(\xi)$ of the same size are called \emph{left unimodularly equivalent} if there exists a unimodular matrix $U(\xi)$ such that $Q(\xi)=U(\xi)R(\xi)$. In the next two definitions $\mathcal{J}$ is assumed to be a subset of $\{1,\dots,N\}$.

\begin{definition}
	Define $R_\mathcal{J}(\xi)$ as an $N\times\|\mathcal{J}\|$ matrix that consists of the $i$-th columns of $R(\xi)$ where $i\in\mathcal{J}$.
\end{definition}

\begin{definition}
	Define $y_\mathcal{J}$ as the signal that consists of the $i$-th components of $y$ where $i\in\mathcal{J}$; thus $y_{\bar{\mathcal{J}}}$ is the signal that consists of the $i$-th components of $y$ where $i\notin\mathcal{J}$.
\end{definition}

The following theorem is a reformulation of a result in~\cite{chongK2016mtns,Chong2016cdc}.
\begin{thm} [Security index calculation]\label{Theorem:kernel_sec_ind}
	Consider a system $\Sigma$ whose behavior $\mathcal{B}$ is non-zero and given by \eqref{eq:system}, where $R(\xi)$ has full rank. Then 
	\begin{equation}
	\delta(\Sigma) = L+1, \label{eq:kernel_sec_ind}
	\end{equation}
	where $L$ is the largest integer such that for any subset $\mathcal{J}\subseteq\{1,\dots,N\}$ of cardinality $L$, the $N\times L$ matrix $R_{\mathcal{J}}(\xi)$ is left unimodular.
\end{thm} 
\begin{proof}
	Clearly, there exists a subset $\mathcal{J}\subseteq\{1,\dots,N\}$ of cardinality $L+1$ such that $R_{{\mathcal{J}}}(\xi)$ is not left unimodular. Thus there exists a nonzero signal $y^\star$ that satisfies $R_{{\mathcal{J}}}(\sigma)y^\star = 0$. Now let $y:\mathbb{Z}_+\rightarrow\mathbb{R}^N$ be the signal satisfying $y_{\mathcal{J}}= y^\star$ and $y_{\bar{\mathcal{J}}}= 0$. Then $y \in \mathcal{B}$ and $\|y\| = \|y^\star\| \leq L+1$. This implies that 
	\begin{equation}
	\delta(\Sigma) \leq L+1. \label{eq:kernel_sec_ind_1}
	\end{equation}
	To prove that also $\delta(\Sigma) \geq L+1$, let $y$ be a signal in $\mathcal{B}$ of weight $\delta(\Sigma)$. Define $\bar{\mathcal{J}} \subset \{1,2, \ldots , N\}$ as the set of cardinality $\delta(\Sigma)$ for which $y_{{\mathcal{J}}} = 0$. Then $R_{\bar{\mathcal{J}}}(\xi) y_{\bar{\mathcal{J}}} = 0$ and because $y_{\bar{\mathcal{J}}} \neq 0$ it follows that $R_{\bar{\mathcal{J}}}(\xi)$ is not left unimodular. This implies that $L< \delta(\Sigma)$. Because of~(\ref{eq:kernel_sec_ind_1}) it follows that equation~(\ref{eq:kernel_sec_ind}) holds.
\end{proof}
\begin{cor} \label{cor:ker_max}(\cite[Cor. IV.6]{Chong2016cdc})
	The system $\Sigma$ in \eqref{eq:system} is maximally secure if and only if all $N\times (N-1)$ submatrices of $R(\xi)$ are left unimodular.
\end{cor}

\section{Canonical kernel representation}\label{sec:sim_tra}
When we describe a system's behavior $\mathcal{B}$ using a kernel representation $R(\sigma)y=0$, the polynomial matrix $R(\xi)$ is not unique. In this section, we recall results from the literature around equivalent kernel representations for $\mathcal{B}$. We then use this to single out a canonical form of $R(\xi)$ which is vital for our results on attack correction.
\begin{thm}\label{Theorem:uni_man}(e.g. \cite[Theorem 3.9]{kuijbook94})
	Consider two systems $\Sigma$ and $\Sigma'$ whose behaviors $\mathcal{B}$ and $\mathcal{B}'$ are given by $\mathcal{B} = \{y|R(\sigma)y = 0\}$ and $\mathcal{B}' = \{y'|R'(\sigma)y' = 0\}$, respectively. Assume that $R(\xi)$ and $R'(\xi)$ are square matrices of the same size, then $\mathcal{B} = \mathcal{B}'$ if and only if $R(\xi)$ and $R^{'}(\xi)$ are left unimodularly equivalent.
\end{thm}

Minimal lag kernel representations (i.e., where the row degrees of $R(\xi)$ are minimal, namely the system's observability indices) are relevant to many problems in systems and control. However, it turns out that minimal lag representations do not lend themselves well to the design of attack correction methods. Instead, we find that a different canonical form needs to be used, namely the Kronecker-Hermite form. We recall this theory in the next theorem and more generally in Theorem~\ref{canon_general}. In the next section we will see that it serves as an important tool for our attack correction method.
\begin{thm}\label{Theorem_canonical}\textbf{\textup{(Kronecker-Hermite canonical kernel representation of a maximally secure system)}}
	Let $R(\xi)$ be a $N\times N$ polynomial matrix whose determinant is non-zero. Assume that all $N\times (N-1)$ submatrices of $R(\xi)$ are left unimodular. Then there exists a unimodular matrix $U(\xi)$ such that
	\begin{equation}\label{eq:des_form}
	U(\xi){R}(\xi) = \left[
	\begin{array}{c|c}
	\mathbb{I}_{N-1}&\begin{array}{c}
	-c_{1}(\xi)\\
	-c_{2}(\xi)\\
	\vdots\\
	-c_{N-1}(\xi)
	\end{array}\\
	\hline
	\begin{array}{ccc}
	0&\dots&0
	\end{array}&a(\xi)
	\end{array}
	\right],
	\end{equation}
	where $c_j(\xi)$ is coprime with $a(\xi)$ and $\DEG c_j(\xi) < \DEG a(\xi)$ for all $j\in\{1,...,N-1\}$ .
\end{thm}
\begin{proof}
	It follows from e.g. Theorem B.1.1 in \cite{poldermanW97} that there exists a unimodular matrix $U_0(\xi)$ such that $U_0(\xi)R(\xi)$ is an upper triangular polynomial matrix, say $R_0(\xi)$, written as 
	\begin{align*}
	&R_0(\xi) = \\
	&\begin{bmatrix}\label{eq:upper_tri}
	a_1(\xi)&b_{12}(\xi)&b_{13}(\xi)&\dots&b_{1(N-1)}(\xi)&b_{1N}(\xi)\\
	0&a_2(\xi)&b_{23}(\xi)&\dots&b_{2(N-1)}(\xi)&b_{2N}(\xi)\\
	0&0&a_3(\xi)&\dots&b_{3(N-1)}(\xi)&b_{3N}(\xi)\\
	\vdots&\vdots&\vdots&\ddots&\vdots&\vdots\\
	0&0&0&\dots&a_{N-1}(\xi)&b_{(N-1)N}(\xi)\\
	0&0&0&\dots&0&a_N(\xi)	
	\end{bmatrix}.
	\end{align*}
	Since all $N\times (N-1)$ submatrices of $R(\xi)$ are left unimodular, it follows that in particular the matrix formed by the first $N-1$ columns is left unimodular.  This implies that all its diagonal elements are nonzero constants. Without restrictions, $R_0(\xi) $ can then be written as
	\begin{equation}
	\begin{bmatrix}\label{eq:upper_tri_1}
	1&b_{12}(\xi)&b_{13}(\xi)&\dots&b_{1(N-1)}(\xi)&b_{1N}(\xi)\\
	0&1&b_{23}(\xi)&\dots&b_{2(N-1)}(\xi)&b_{2N}(\xi)\\
	0&0&1&\dots&b_{3(N-1)}(\xi)&b_{3N}(\xi)\\
	\vdots&\vdots&\vdots&\ddots&\vdots&\vdots\\
	0&0&0&\dots&1&b_{(N-1)N}(\xi)\\
	0&0&0&\dots&0&a_N(\xi)	
	\end{bmatrix} .
	\end{equation}
	
	It has been shown in e.g.~\cite[Theorem 2.40]{Fuhrmann15}, ~\cite[Theorem 7.5]{colonius01} or~\cite{hinrichsen83} that there exists a unimodular matrix $U_1(\xi)$ such that $U_1(\xi)R_0(\xi)$ is in Kronecker-Hermite canonical form as in \eqref{eq:des_form} with $\DEG a(\xi)>\DEG c_j(\xi)$. Finally, according to Corollary \ref{cor:ker_max}, the matrix in~\eqref{eq:des_form} has the property that all its $N\times(N-1)$ submatrices are left unimodular. It can be easily checked that this implies that $c_j(\xi)$ is coprime with $a(\xi)$, i.e., $\text{GCD}(c_j(\xi),a(\xi)) = 1$ for all $j\in\{1,...,N-1\}$, and this completes the proof. 
\end{proof}

\section{Attack detection and correction algorithms} \label{sec:attack_detect} 
In this section, we propose methods to achieve attack detection and correction. Recall (\ref{eq:system}) and (\ref{eq:system_attacked}), the output generated by the attacked system $\Sigma_\mathcal{A}$ is $r = y+\eta$ where $y$ is the attack-free sensor output of the system $\Sigma$. In linear coding theory, syndrome computation is an effective method for error detection, see e.g.~\cite[Ch. 7]{Proakis08}. Similarly we will work with a signal that we call the ``residual signal", defined as $s = R(\sigma)r$ to perform attack detection (Section \ref{sec:Attack_Detection_Alg}). To achieve attack correction (Sections \ref{sec:Correction_MS} and \ref{Correction_GC}), we use a majority vote rule reminiscent of decoding techniques such as in~\cite{reedS60}.

\subsection{Attack detection}
\label{sec:Attack_Detection_Alg}

For attack detection, we propose Algorithm \ref{alg:att_det} below. 
\begin{algorithm}
	\caption{Attack detection}\label{alg:att_det}
	\begin{algorithmic}[1]
		\Procedure{}{$R(\xi),r,\eta$} \newline\Comment{Given $R(\xi)$ and $r$, detect whether $\eta$ is the zero signal}.
		\State Calculate $s = R(\sigma)r$.
		\If{$s=0$} \space decide no attack, i.e., $\eta = 0$.
		\Else \space 
		decide attack occurred, i.e., $\eta \neq 0$.
		\EndIf
		\EndProcedure
	\end{algorithmic}
\end{algorithm}

\begin{thm}[Attack detection]\label{Theorem:Attack_Detection}
	Consider a system given by \eqref{eq:system}. Let $r=y+\eta$ be a received signal with $y \in \mathcal{B}$. Then the residual signal $s=0$ if and only if $\eta \in \mathcal{B}$. Thus Algorithm~\ref{alg:att_det} gives the correct result if $\eta$ is detectable.
\end{thm}
\begin{proof}
	This follows straightforwardly from
	\begin{equation*}
	s=R(\sigma ) r=R(\sigma ) (y+\eta )=R(\sigma ) y+R(\sigma )\eta  =R(\sigma )\eta .
	\end{equation*}	
\end{proof}

\subsection{Attack correction for a maximally secure system}
\label{sec:Correction_MS}
In this section we show how the Kronecker-Hermite canonical form kernel representation of the previous section can be used to perform attack correction for a maximally secure system. Without loss of generality, we assume that $R(\xi)$ is in the Kronecker-Hermite canonical form \eqref{eq:des_form}. Thus the system is given by
\begin{equation}
\label{eq:ms_system_equation}
\begin{bmatrix}
\mathbb{I}_N\\0
\end{bmatrix}y = \begin{bmatrix}
c_1(\sigma)\\\vdots\\c_{N-1}(\sigma)\\1\\a(\sigma)
\end{bmatrix}y_N .
\end{equation}
Before defining our method of attack correction, we need the following definitions and computations:
\begin{itemize}
	\item Define polynomials  $p_j(\xi)$ and $q_j(\xi)$ satisfying 
	\begin{equation}\label{eq:coprime}
	\begin{bmatrix}
	p_j(\xi)&q_j(\xi)
	\end{bmatrix}\begin{bmatrix}
	c_j(\xi)\\a(\xi)
	\end{bmatrix} = 1,\forall j\in\{1,...,N-1\}.
	\end{equation}
	%	\item Define signal $s_{j,N}$s satisfying
	%	\begin{equation}\label{eq:new_res}
	%	s_{j,N} := \begin{bmatrix}
	%	p_j(\sigma)&q_j(\sigma)
	%	\end{bmatrix}\begin{bmatrix}
	%	s_j\\s_{N}
	%	\end{bmatrix}, \forall j\in\{1,...,N-1\}.
	%	\end{equation}
	Note that the existence of $p_j(\xi)$ and $q_j(\xi)$ follows from the fact that $\text{GCD}(c_j(\xi),a(\xi)) = 1$; the Extended GCD Algorithm (e.g. Ch 4.2 in \cite{Stallings11}) can be used to find $p_j(\xi)$ and $q_j(\xi)$.
	\item The majority vote function over a set of signals $\{ v_1,v_2, \dots , v_L\}$, denoted by $\text{Maj}\{ v_1,v_2, \dots , v_L\}$, is defined to be the most frequently occurring signal in the set of $v_j$'s.
\end{itemize}

\begin{algorithm}
	\caption{Attack correction for a maximally secure system given by \eqref{eq:ms_system_equation}}\label{alg:att_corr}
	\begin{algorithmic}[1]
		\Procedure{}{$a(\xi), c_1(\xi),\ldots , c_{N-1}(\xi),r,\hat y$} \newline\Comment{Given $a(\xi)$, $c_j(\xi)$'s and $r$, compute $\hat y$}.
		\State Calculate 
		\begin{equation}\label{eq_majority}
		\hat y_N = \text{Maj}\{p_1(\sigma )r_1 , p_2(\sigma)r_2 , \ldots , p_{N-1}(\sigma) r_{N-1} , r_N \} ,
		\end{equation}
		where $p_j(\xi)$ is defined as in \eqref{eq:coprime}.
		\State $\hat y_j = c_j(\sigma ) \hat y_N$ for $j=1, 2, \ldots , N-1$.
		\State\textbf{return} $\hat{y} = \begin{bmatrix}
		\hat{y}_1&\hat{y}_2&\dots&\hat{y}_N
		\end{bmatrix}$.
		\EndProcedure
	\end{algorithmic}
\end{algorithm}
\begin{thm}\label{thm_maxSecureMain}
	Consider a maximally secure system $\Sigma$ given by \eqref{eq:system}. Let the received signal $r$ be input to Algorithm~\ref{alg:att_corr}. Assume that $r=y+\eta$ with $y \in \mathcal{B}$ and $\|{\eta}\|<N/2$. Then the output $\hat y$ of Algorithm~\ref{alg:att_corr} equals $y$.
\end{thm}
\begin{proof}
	We have
	\[
	\begin{bmatrix}
	1\\0
	\end{bmatrix} y_j = \begin{bmatrix}
	c_j(\sigma)\\a(\sigma)
	\end{bmatrix}y_N,
	\]
	so that it follows from \eqref{eq:coprime} that $p_j (\sigma ) y_j = y_N$ for $j=1, 2, \ldots , N$ (here we define $p_N(\xi) \equiv 1$). Now
	\[
	p_j (\sigma ) r_j = p_j (\sigma ) (y_j + \eta_j) = y_N + p_j (\sigma ) \eta_j 
	\]
	for $j=1, 2, \ldots , N$. Since $\|{\eta}\|<N/2$ it follows that \eqref{eq_majority} computes $\hat y_N = y_N$. Consequently $y_j$ can be found from $y_j = c_j(\sigma ) y_N$ for $j=1, 2, \ldots , N-1$ and this proves the theorem.
	
\end{proof}
Note that the results in this subsection require the system to be maximally secure. The next subsection deals with the general case where systems are not necessarily maximally secure. 
\subsection{Attack correction for the general case}
\label{Correction_GC}
%%%%%%%%%%%%%%%%%%%%%%%%%%%%%%%%%%%%%%
\begin{thm}\label{canon_general}\textup{\textbf{(Kronecker-Hermite canonical kernel representation of $R(\xi)$---general case)}} Let $R(\xi)$ be a $N\times N$ polynomial matrix whose determinant is nonzero. Let $L$ be the largest integer such that, for any subset $\mathcal{J}\subseteq\{1,\dots,N\}$ of cardinality $L$, the $N\times L$ matrix $R_{\mathcal{J}}(\xi)$ is left unimodular. Then there exists a unimodular matrix $U(\xi)$ such that
	\begin{equation}\label{eq:des_form_general}
	U(\xi){R}(\xi) = \left[
	\begin{array}{c|c}
	\mathbb{I}_{L}&\begin{array}{c}
	-M_1(\xi)
	\end{array}\\
	\hline
	\begin{array}{ccc}
	0&\dots&0
	\end{array}&D(\xi)
	\end{array}
	\right].
	\end{equation}
	where $D(\xi)$ is an upper triangular matrix and the degree of the diagonal entities of $D(\xi)$ denoted as $\DEG d_{ii}(\xi)$ for $i\in\{1,...,N-L\}$, is strictly the highest within the corresponding column of \eqref{eq:des_form_general}.
\end{thm}
\begin{proof}
	The proof follows the same reasoning as the proof of Theorem~\ref{Theorem_canonical}, but replacing $N-1$ by $L$.
\end{proof}
Combining Theorem~\ref{Theorem:kernel_sec_ind} and Theorem~\ref{canon_general}, it follows that the signals $y$ in the behavior of a system $\Sigma$ with security index $\delta$ are given by the following representation
\begin{equation}\label{eq:MAdelta}
\begin{bmatrix}
\mathbb{I}_{\delta-1} & 0\\0&\mathbb{I}_{N-\delta +1}\\0 & 0 
\end{bmatrix}y = \begin{bmatrix}
M_1(\sigma)\\\mathbb{I}_{N-\delta +1}\\D(\sigma)
\end{bmatrix}\ell ,
\end{equation}
where the signal $\ell$ is an auxiliary signal that can be interpreted as a ``state signal'' that drives the system's behavior. In the representation~\eqref{eq:MAdelta} the signal $\ell$ simply coincides with the last $N-\delta +1$ components of $y$.

In fact, the above representation~\eqref{eq:MAdelta} is a special case of a more general representation~\cite{willems93,kuijaut95}, given as 
\begin{equation}\label{eq:MD}
\begin{bmatrix}
\mathbb{I}_N\\0 
\end{bmatrix}y = \begin{bmatrix}
M(\sigma)\\D(\sigma)
\end{bmatrix}\ell ,
\end{equation}
where $\ell: \mathbb{Z}_+\rightarrow\mathbb{R}^m$, $M(\xi)$ is a $N\times m$ polynomial matrix and $D(\xi)$ is a $m\times m$ polynomial matrix, for some integer $m$. We make the observability assumption that the $(N+m)\times m$ polynomial matrix $\begin{bmatrix}
M(\xi)\\D(\xi)
\end{bmatrix}$ is left unimodular, noting that this clearly holds for the above representation~\eqref{eq:MAdelta}. Note that the representation~\eqref{eq:des_form} of the previous subsection is a special case of~\eqref{eq:MD}, namely 
\[
M(\xi ) := \begin{bmatrix}
c_1(\xi)\\\vdots \\ c_{N-1}(\xi)\\1 \end{bmatrix} \mbox{ and } D(\xi) := a(\xi) .
\]

As a first step towards attack correction in the general case, we express the system's security index in terms of the polynomial matrices $M(\xi)$ and $D(\xi)$ of the general representation \eqref{eq:MD}. 
\begin{thm}\label{Theorem:MD_sec_ind}
	Consider a system $\Sigma$ whose behavior $\mathcal{B}$ is nonzero and given by \eqref{eq:MD}. Then 
	\begin{equation}
	\delta(\Sigma) = N+1 - \tilde L, \label{eq:MD_sec_ind}
	\end{equation}
	where $\tilde L$ is the smallest integer such that for any subset $\mathcal{J}\subseteq\{1,\dots,N\}$ of cardinality $\tilde L$, the $(\tilde L+m)\times m$ matrix $\begin{bmatrix}
	M_\mathcal{J}(\xi)\\D(\xi)
	\end{bmatrix}$ is left unimodular.
\end{thm} 
\begin{proof}
	Clearly, there exists a subset $\mathcal{J}\subseteq\{1,\dots,N\}$ of cardinality $\tilde L-1$ such that $\begin{bmatrix}
	M_\mathcal{J}(\xi)\\D(\xi)
	\end{bmatrix}$ is not left unimodular. Thus there exists a nonzero signal $\ell^\star$ that satisfies $\begin{bmatrix}
	M_\mathcal{J}(\sigma)\\D(\sigma)
	\end{bmatrix}\ell^\star = 0$. Now consider the signal $y$ defined as $y:=M(\sigma) \ell^\star$. Clearly $\|y\| \leq N-(\tilde L-1)=N-\tilde L +1$. This implies that
	\begin{equation}
	\delta(\Sigma) \leq N- \tilde L+1. \label{eq:MD_sec_ind_1}
	\end{equation}
	To prove that also $\delta(\Sigma) \geq N- \tilde L+1$, let $y^\star$ be a signal in $\mathcal{B}$ of weight $\delta(\Sigma)$. Thus there exists a nonzero signal $\ell^\star$ such that
	\begin{equation}
	\begin{bmatrix}
	\mathbb{I}_N\\0 
	\end{bmatrix}y^\star = \begin{bmatrix}
	M(\sigma)\\D(\sigma)
	\end{bmatrix}\ell ^\star.
	\end{equation}
	Define $\bar{\mathcal{J}} \subset \{1,2, \ldots , N\}$ as the set of cardinality $\delta(\Sigma)$ for which $y^\star_{{\mathcal{J}}} = 0$. Then $\begin{bmatrix}
	M_{{\mathcal{J}}}(\sigma)\\D(\sigma)
	\end{bmatrix}\ell ^\star = 0$ and because $\ell ^\star \neq 0$, it follows that $\begin{bmatrix}
	M_{{\mathcal{J}}}(\sigma)\\D(\sigma)
	\end{bmatrix}$ is not left unimodular. This implies that $\tilde L\geq N -  \delta(\Sigma) +1$. Because of~(\ref{eq:MD_sec_ind_1}), it follows that equation~(\ref{eq:MD_sec_ind}) holds. 
\end{proof}

Before defining our method of attack correction, we need several preliminary computations. Let $\mathcal{J}$ be a subset of $\{1,\dots,N\}$ of cardinality $N+1 -\delta$. Suppose that the matrix $\begin{bmatrix}
M_{\mathcal{J}}(\xi)\\D(\xi)
\end{bmatrix} $ is left unimodular. Define 
polynomial matrices $P^{\mathcal{J}}(\xi)$ and $Q^{\mathcal{J}}(\xi)$ such that 
\begin{equation}\label{eq:coprime_general}
\begin{bmatrix}
P^{\mathcal{J}}(\xi)&Q^{\mathcal{J}}(\xi)
\end{bmatrix}\begin{bmatrix}
M_{\mathcal{J}}(\xi)\\D(\xi)
\end{bmatrix} = \mathbb{I}_{N+1 -\delta} .
\end{equation}
%%%%
\begin{algorithm}
	\caption{Attack correction for general system given by \eqref{eq:MD}}\label{alg:att_corr_general}
	\begin{algorithmic}[1]
		\Procedure{}{$M(\xi),D(\xi),\delta,r,\hat y$} \newline\Comment{Given $M(\xi),D(\xi)$, $\delta$ and $r$, compute $\hat y$}.
		\State Calculate 
		\begin{equation}\label{eq_majority_general}
		\hat \ell = \text{Maj}\{P^{\mathcal{J}}(\sigma )r_{\mathcal{J}}  \} ,
		\end{equation}
		where the majority vote is taken over all subsets ${\mathcal{J}}$ of cardinality $N+1 -\delta$ and $P^\mathcal{J}(\xi)$ is defined as in \eqref{eq:coprime_general}.
		\State $\hat y = M(\sigma) \hat \ell$.
		\State\textbf{return} $\hat y$.
		\EndProcedure
	\end{algorithmic}
\end{algorithm}

In Algorithm \ref{alg:att_corr_general} we can interpret each computation $P^{\mathcal{J}}(\sigma )r_{\mathcal{J}} $ as an exact observer (see proof of Theorem \ref{Theorem:gen_case_correction}) that produces an estimated signal $\hat \ell$ from the received signal $r$. The algorithm directs us to take a majority vote of all such observer outcomes and to declare this signal to be the correct signal $\ell$. As we will see below, this algorithm and the next theorem are the main results of this paper. We note that the next theorem requires a nontrivial proof, reminiscent of majority vote proofs in the classical coding literature, such as~\cite{reedS60}.
\begin{thm}
	\label{Theorem:gen_case_correction}
	Consider a system $\Sigma$ given by \eqref{eq:system}; denote its security index by $\delta$. Let the received signal $r$ be input to Algorithm~\ref{alg:att_corr_general}. Assume that $r=y+\eta$ with $y \in\mathcal{B}$ and $\|{\eta}\|<\delta/2$.
	Then the output $\hat y$ of Algorithm~\ref{alg:att_corr_general} equals $y$.
\end{thm}
\begin{proof}
	Let's denote $\| \eta \| =$ the number of attacked sensors by $t$, thus $2t < \delta$. Let ${\mathcal{J}} $ be a subset of cardinality $N+1 -\delta $ from the set of unattacked sensors. Then
	\begin{equation}\label{M_J}
	\begin{bmatrix}
	\mathbb{I}_{N+1-\delta}\\0 
	\end{bmatrix}r_{\mathcal{J}} =\begin{bmatrix}
	\mathbb{I}_{N+1-\delta}\\0 
	\end{bmatrix}y_{\mathcal{J}} = \begin{bmatrix}
	M_{\mathcal{J}}(\sigma) \\D(\sigma)
	\end{bmatrix}\ell,
	\end{equation}
	where $\ell$ is the correct signal. We first show that \eqref{eq_majority_general} is well defined. Because of Theorem \ref{Theorem:MD_sec_ind}, the matrix $\begin{bmatrix}
	M_\mathcal{J}(\xi)\\D(\xi)
	\end{bmatrix}$ is left unimodular, so that matrices $P_\mathcal{J}(\xi)$ and $Q_\mathcal{J}(\xi)$ can be found such that \eqref{eq:coprime_general} holds. Using \eqref{eq:coprime_general}, it follows from (\ref{M_J}) that $\hat \ell = P^{\mathcal{J}}(\sigma )r_{\mathcal{J}}$ equals the correct signal $\ell$. There are ${N-t \choose N+1 -\delta }$ ways to choose a subset ${\mathcal{J}} $ of cardinality $N+1 -\delta $ from the set of unattacked sensors. Each of these choices leads to $\hat \ell$ as the correct signal $\ell$. 
	\newline
	Next, let's consider a subset ${\mathcal{J}^\star} $ of cardinality $N+1 -\delta $ that leads to a signal $\hat \ell$ that is incorrect, say $\ell^\star \neq \ell$. Clearly ${\mathcal{J}^\star} $ must involve one or more attacked sensors. Since ${\mathcal{J}^\star} $ has only $N+1 -\delta $ elements, it follows that all unattacked sensors that are involved in ${\mathcal{J}^\star} $ fit into a certain subset, say ${\mathcal{I}} $ of $N-\delta$ unattacked sensors. Now define the set ${\mathcal{I}}^\star $ as the union of ${\mathcal{I}} $ and all $t$ attacked sensors. Then a set ${\mathcal{J}} $ of cardinality $N+1 -\delta $ that is a subset of ${\mathcal{I}}^\star $ may lead to the same incorrect $\ell^\star$. Let's now consider a set ${\mathcal{\tilde J}} $ of cardinality $N+1 -\delta $ that 
	leads to the incorrect signal $\ell^\star$ but that is not a subset of ${\mathcal{I}}^\star $. We note that ${\mathcal{\tilde J}} $ must then involve an unattacked sensor outside of ${\mathcal{I}} $. Thus there exist more than $N-\delta$ unattacked sensors that lead to the same signal $\ell^\star$. Since any set of $N-\delta +1$ unattacked sensors lead to the correct signal $\ell$, it follows that 
	$\ell^\star$ must be the correct signal $\ell$, which is a contradiction. We conclude that ${\mathcal{\tilde J}} $ does not lead to $\ell^\star$.
	\newline
	From this reasoning, it follows that there are at most ${N-\delta + t \choose N+1 -\delta }$ choices of ${\mathcal{J}} $ that lead to the same incorrect $\ell^\star$. Now recall from the first part of this proof that there are at least ${N-t \choose N+1 -\delta } $ choices of ${\mathcal{J}} $ that lead to the correct signal $\ell$. Since $2t < \delta$ we have
	\[
	{N-t \choose N+1 -\delta } > {N-\delta + t \choose N+1 -\delta } ,
	\]
	so that the majority vote in Algorithm~\ref{alg:att_corr_general} leads to the correct $\ell$. As a result, $\hat y=M(\sigma ) \hat \ell = M(\sigma ) \ell $ equals the correct signal $y$ and this proves the theorem. 
\end{proof}

\subsection{Comparison with the literature}

The attack correction method for the general case of Subsection \ref{Correction_GC} is useful because the representation \eqref{eq:MD} is very general. For example, a special case of \eqref{eq:MD} is
\[
M(\xi ) := \begin{bmatrix}
c_1(\xi)\\\vdots \\ c_{N}(\xi) \end{bmatrix} \mbox{ and } D(\xi) := a(\xi) ,
\]
and we see that Lemma IV-8 of~\cite{Chong2016cdc} then follows immediately from the above theorem. 
\newline
Another special case is a state space representation---then $M(\xi)$ equals a constant $N\times n$ ``output matrix" $C$ and $D(\xi) = \xi \mathbb{I}_n - A$, where $A$ is a constant ``state transition matrix". As this is the type of representation that is considered in e.g.~\cite{ChongWakaikiHespanhaACC15,shoukryT2016,shoukryCWNVSHT2016,fawziTD14}, we are now in a position to compare our results with the literature in this area. In contrast to our representation-free definition of ``security index", much of the literature implicitly defines a similar notion in terms of the matrices $A$ and $C$ of the state space representation. For example~\cite{ChongWakaikiHespanhaACC15} defines {\em $M$-attack observability} of a state space representation $(A,C)$ in terms of the matrices $A$ and $C$. In terms of our notion of security index $\delta(\Sigma)$, any observable state space representation $(A,C)$ is $\lfloor \delta(\Sigma)/2 \rfloor$-attack observable. Alternatively, in the terminology of~\cite{shoukryT2016,shoukryCWNVSHT2016}, this is phrased as $(\delta(\Sigma) -1)$-{\em sparse attack observability}.

Comparing our Algorithm~\ref{alg:att_corr_general} to the noise-free attack correction method of ~\cite[Section III-A]{ChongWakaikiHespanhaACC15} and~\cite[Section III-A]{shoukryCWNVSHT2016}, we see that our Algorithm~\ref{alg:att_corr_general} is simpler with lower computational complexity, as it requires fewer observers and fewer observer comparisons.

\section{Examples}\label{sec:example}
We illustrate the workings of the proposed attack correction method with two examples. In the first example the $A$ matrix is in companion form so as to be able to illustrate our theory with very simple observers. The second example has an $A$ matrix that is not in companion form. 
\subsection{Example 1}
Consider a discrete-time LTI system given by a state space representation $(A,C)$ with
\begin{equation*}
A = \begin{bmatrix}
0&1&0\\
0&0&1\\
\frac{1}{2}&-\frac{3}{2}&\frac{3}{2}
\end{bmatrix},\quad C = \mathbb{I}_3 ,\quad y(0) = \begin{bmatrix}   1 \\ 1 \\ 1 \end{bmatrix}.
\end{equation*}
Note that this system is marginally stable with three distinct eigenvalues: $\lambda_1 = 1\angle60^\circ$, $\lambda_2 = 1\angle-60^\circ$ and $\lambda_3 = \frac{1}{2}$. Its Kronecker-Hermite canonical kernel representation can be computed as
\begin{equation*}
R(\xi) = \begin{bmatrix}
1&0&-6\xi^2+7\xi-6\\
0&1&-2\xi^2+3\xi-3\\
0&0&\xi^3-\frac{3}{2}\xi^2+\frac{3}{2}\xi-\frac{1}{2}
\end{bmatrix}
\end{equation*}
Using Corollary \ref{cor:ker_max}, we find that the system is maximally secure, i.e., $\delta(\Sigma) = 3$. According to Theorems \ref{Theorem:detect} and \ref{Theorem:Correctability} the system can detect up to 2 sensor attacks and it can correct any single sensor attack. The value of observer $p_j(\xi)$s and corresponding $q_j(\xi)$s that satisfy equation \eqref{eq:coprime} can be computed as follows:
\begin{align*}
p_1(\xi) &=\xi^2, \;\;\; q_1(\xi)= -6\xi-2\\
p_2(\xi) &=\xi, \;\;\;  q_2(\xi)=-2
\end{align*}
Algorithm \ref{alg:att_corr} yields the following attack correction outcomes:
\begin{equation}
\label{eq:exam_correction}
\begin{split}
\hat{y}_3 &= \text{Maj}\{\sigma^2r_1,\sigma r_2,r_3\}\\
\hat{y}_1 &= (6\sigma^2-7\sigma+6)\hat{y}_3\\
\hat{y}_2 &= (2\sigma^2-3\sigma+3)\hat{y}_3
\end{split}
\end{equation}
Figure \ref{fig:sys_mod} shows the block diagram of the attack correction method. 
\begin{figure}[h]
	\centering
	\includegraphics[width=0.48\textwidth]{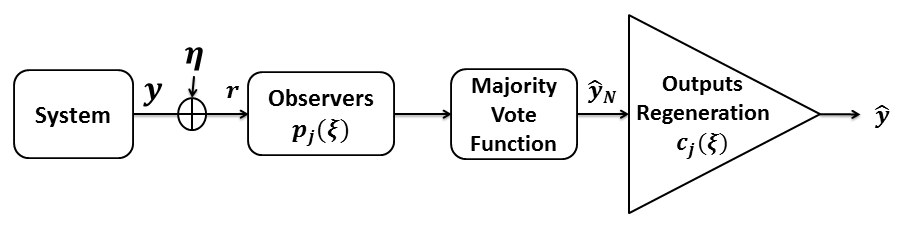}\\
	\caption{Block diagram for attack correction}
	\label{fig:sys_mod}
\end{figure}

Next, we conduct a simulation to illustrate the workings of the attack correction method for this example. First, we specify a single sensor attack signal on sensor 3, as illustrated in Figure \ref{fig:attack_signal}. Here the attack signal is generated as an i.i.d. random sequence using a uniform distribution on the interval $(-1,1)$. 

\begin{figure}[h]
	\centering
	\includegraphics[width=0.48\textwidth,height=0.16\textwidth]{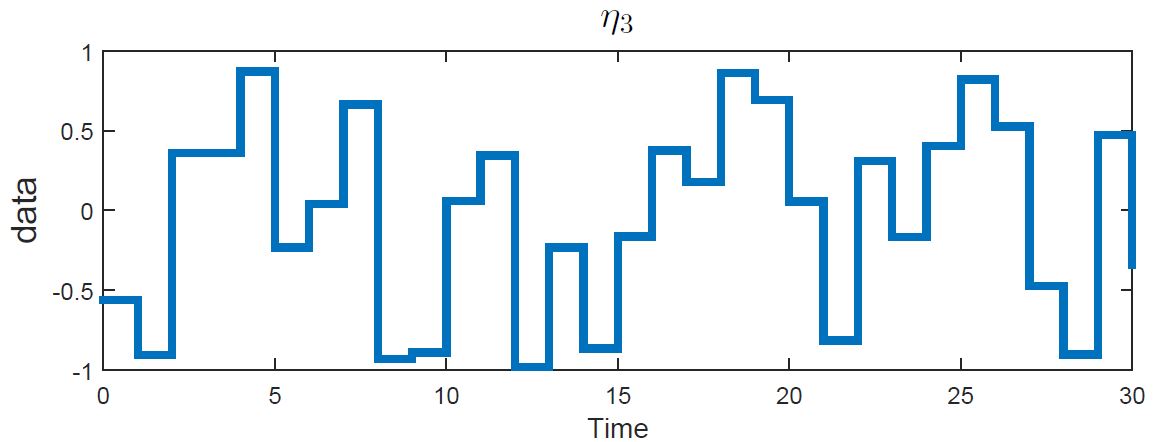}\\
	\caption{Attack signal on sensor 3}
	\label{fig:attack_signal}
\end{figure}

Figure \ref{fig:Theo_Obs} shows the observer outputs. It can be seen that apart from the first 2 time instants, observer outputs 1 and 2 are identical and thus selected by the majority vote function for further processing. The inequality for the first 2 time instants is caused by the latency in the observers, more specifically the latency equals $\max_j \DEG p_j(\xi) = 2$.

\begin{figure}[h]
	\centering
	\includegraphics[width=0.5\textwidth]{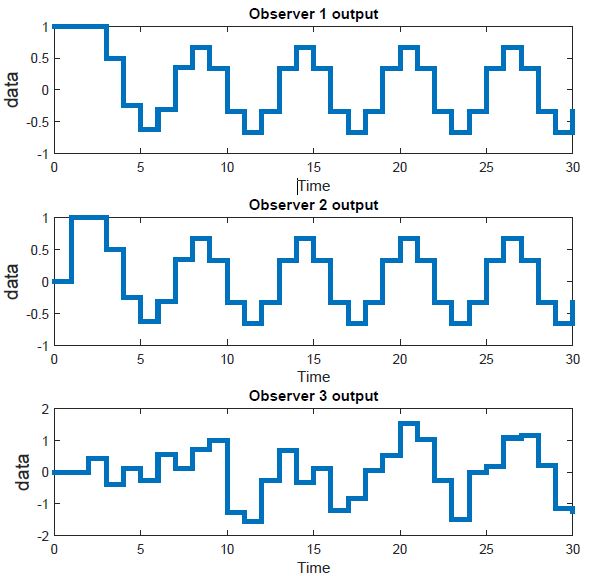}\\
	\caption{Observer $p_j(\xi)$ outputs}
	\label{fig:Theo_Obs}
\end{figure}

Figure \ref{fig:Theo_Output} shows the signals for $\hat{y}_i-y_i$. It can be seen that apart from the first 4 time instants, the difference between the corrected output signals and the attack free system output is negligible, meaning that Algorithm \ref{alg:att_corr} indeed reconstructs the attacked 3rd output signal correctly in accordance with Theorem~\ref{thm_maxSecureMain}. The inequality at the first 4 time instants is caused by the latency in observers $p_j(\xi)$s and output regeneration shift operators $c_j(\xi)$, more specifically the latency equals $\max_i \deg p_i(\xi)+\max_j \deg c_j(\xi) = 4$.

\begin{figure}[h]
	\centering
	\includegraphics[width=0.5\textwidth]{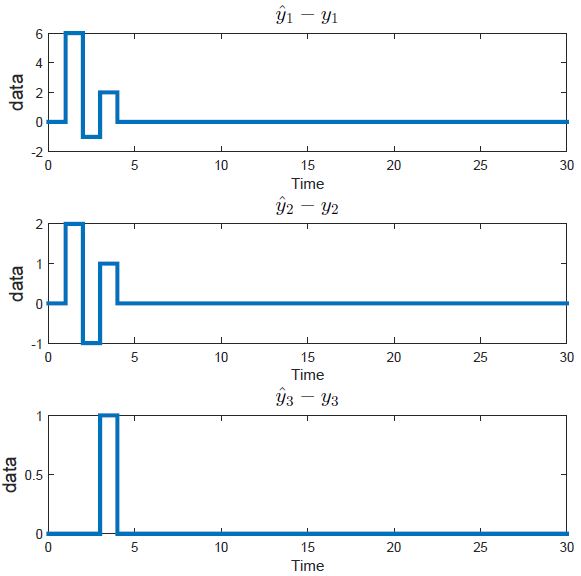}\\
	\caption{Comparison for different outputs}
	\label{fig:Theo_Output}
\end{figure}
\subsection{Example 2}
Consider a 6-output stable voltage-source converter \cite{Huerta2012} that is modelled as a discrete-time LTI system.
After discretization by means of the zero-order hold method with sampling period $T_s = 200\mu s$, the resulting $R(\xi)$ matrix is given by
\begin{equation}\label{eq:example_R_matrix}
R(\xi) = \xi\mathbb{I}_6-\text{exp}(\tilde AT_s),
\end{equation}
where $\tilde A$ is given by
\begin{equation}
\label{eq:system_matrix}
\tilde A = \begin{bmatrix}
-\frac{R_1}{L_1}&\omega_1&0&0&-\frac{1}{L_1}&0\\
-\omega_1&-\frac{R_1}{L_1}&0&0&0&-\frac{1}{L_1}\\
0&0&-\frac{R_2}{L_2}&\omega_1&\frac{1}{L_2}&0\\
0&0&-\omega_1&-\frac{R_2}{L_2}&0&\frac{1}{L_2}\\
\frac{1}{C_0}&0&-\frac{1}{C_0}&0&0&\omega_1\\
0&\frac{1}{C_0}&0&-\frac{1}{C_0}&-\omega_1&0
\end{bmatrix},
\end{equation}
with parameters shown in Table \ref{tab:parameters}.
\begin{table}[H]
	\centering
	\caption{System main parameters}
	\label{tab:parameters}
	\begin{tabular}{|l|l|l|l|l|l}
		\hline
		$L_1$ & 4.3mH  & $R_2$    & 67.3m$\Omega$\\ \hline
		$R_1$ & 83.1m$\Omega$ & $C_0$    & 18$\mu$F   \\ \hline
		$L_2$ & 2.4mH  & $\omega_1$ & 100$\pi$Rad/s  \\ \hline
	\end{tabular}
\end{table}

The Kronecker-Hermite canonical form of $R(\xi)$ can be computed as follows, where we denote $a\times 10^b$ as $aeb$.
\begin{equation*}
\begin{bmatrix}
1&0&0&0&0&\begin{split}
7.4e2\xi^5-1.8e3\xi^4+2.9e3\xi^3\\-2.9e3\xi^2+1.8e3\xi-7.3e2
\end{split}\\
\\
0&1&0&0&0&\begin{split}
94\xi^5-2.7e2\xi^4+4.3e2\xi^3\\-4.8e2\xi^2+2.9e2\xi-1.4e2
\end{split}\\
\\
0&0&1&0&0&\begin{split}
7.4e2\xi^5-1.8e3\xi^4+2.9e3\xi^3\\-2.9e3\xi^2+1.8e3\xi-7.3e2
\end{split}\\
\\
0&0&0&1&0&\begin{split}
94\xi^5-2.7e2\xi^4+4.3e2\xi^3\\-4.8e2\xi^2+2.9e3\xi-1.4e2
\end{split}\\
\\
0&0&0&0&1&\begin{split}
4.7\xi^5-3.2\xi^4+3.3\xi^3\\-2.4\xi^2+1.2\xi-3.3
\end{split}\\
\\
0&0&0&0&0&\begin{split}
3.6e4\xi^6-1.3e5\xi^5+2.3e5\xi^4-2.9\\e5\xi^3+2.3e5\xi^2-1.2e5\xi+3.6e4
\end{split}\\
\end{bmatrix}
\end{equation*}
Using Corollary \ref{cor:ker_max}, we find that the system is maximally secure, i.e., $\delta(\Sigma) = 6$. Furthermore, we compute the observers $p_j(\xi)$s as
\begin{align*}
p_1(\xi) &= 1.3e2\xi^5-2.7e2\xi^4+2.2e2\xi^3-69\xi^2-88\xi+78\\
p_2(\xi) &= 4.1e-2\xi^5-19\xi^4+44\xi^3-65\xi^2+75\xi-35\\
p_3(\xi) &=-72\xi^5+1.5e2\xi^4-1.2e2\xi^3+39\xi^2+49\xi-44\\
p_4(\xi) &=-2.3e{-3}\xi^5+11\xi^4-24\xi^3+36\xi^2-42\xi+19\\
p_5(\xi) &=-4.7\xi^5+3.2\xi^4-3.3\xi^3+2.4\xi^2-1.2\xi+3.3 .
\end{align*}
According to Theorems \ref{Theorem:detect} and \ref{Theorem:Correctability}, any attacks on maximally 5 of the outputs are guaranteed to be detected and attacks on maximally 2 of its outputs are guaranteed to be corrected. Note that for this example it is easy to construct attack scenarios that attack 3 or 4 of the outputs but can still be corrected by our method.

%%%%%%%%%%%%%%%%%%%%%%%%%%%%%%%%%%%%%%
\section{Conclusions and future work}\label{sec:con}

In this paper, we proposed attack detection and correction methods for zero-input discrete LTI systems in the noise-free case. The purpose of this paper is to provide a proof of concept around the application of ideas from error control coding theory to handle attacks on LTI systems. We have shown how the interaction between the inner system dynamics and the sensor placements determines the vulnerability of the system against sensor attacks. We quantified this via the notion of a system security index. We presented detection and correction methods that exploit the known dynamics of the system. In these methods the security index plays a prominent role. 

Future research directions are: build on the discussed results to develop attack detection and correction methods for different system models such as systems with disturbance and noise, multi-input multi-output systems, nonlinear systems and  systems over finite alphabets.

%%%%%%%%%%%%%%%%%%%%%%%%%%%%%%%%%%%%%%%%%%%%%%%%%%%%%%%%%%%%%%%%%%%%%%%%%%%%%%%%
%%%%%%%%%%%%%% %%%%%%%%%%%%%%%%%%%%
\bibliographystyle{plain}
\bibliography{code}
%\begin{thebibliography}{99}
%\end{thebibliography}

\end{document}